\documentclass{llncs}
\usepackage{amsmath}
\usepackage{amssymb}
\usepackage{enumerate}
\usepackage{varioref}
\usepackage{charter,eulervm}
\usepackage{graphicx}
\usepackage{subfig,wrapfig}
\usepackage{boxedminipage}

\newcommand{\es}{\varnothing}

\DeclareMathOperator{\size}{size}

\title{{\sc On Complexities of Minus Domination}}

\author{
 Lu\'erbio~Faria\inst{1}
\and 
 Wing-Kai~Hon\inst{2}
\and
 Ton~Kloks
\and 
 Hsiang-Hsuan~Liu\inst{2}
\and
 Tao-Ming~Wang\inst{3}
\and
 Yue-Li~Wang\inst{4}} 
\institute{
 Instituto de Matem\'atica e Estat\'istica\\
 Universidade do Estado do Rio de Janeiro, Brazil\\
 {\tt luerbio@cos.ufrj.br} 
\and
 Department of Computer Science\\
 National Tsing Hua University, Taiwan\\
 {\tt (wkhon,hhliu)@cs.nthu.edu.tw} 
\and 
 Department of Applied Mathematics\\
 Tunghai University, Taichung, Taiwan\\
 {\tt wang@go.thu.edu.tw}
\and 
 Department of Information Management\\ 
 National Taiwan University of Science and Technology\\ 
 {\tt ylwang@cs.ntust.edu.tw}
}

\begin{document}

\maketitle

\begin{abstract}
A function $f: V \rightarrow \{-1,0,1\}$ 
is a minus-domination function of a graph $G=(V,E)$ 
if the values over the vertices in each closed neighborhood 
sum to a positive number. The weight of $f$ is the sum 
of $f(x)$ over all vertices $x \in V$.  
The minus-domination number $\gamma^{-}(G)$ is the 
minimum weight over all minus-domination functions. 
The size of a minus domination is the 
number of vertices that are assigned $1$. 
In this paper we show that the minus-domination problem 
is fixed-parameter tractable for $d$-degenerate graphs when parameterized 
by the size of the minus-dominating set and by $d$. 
The minus-domination 
problem  
is polynomial for graphs of bounded 
rankwidth and for 
strongly chordal graphs. It is NP-complete for splitgraphs. 
Unless $P=NP$ there is no fixed-parameter algorithm for 
minus-domination.    
\end{abstract}

\section{Introduction}
\label{section intro}

A fresh breeze seems to be blowing through the area of domination 
problems. This research area is aroused anew 
by the recent fixed-parameter 
investigations (see, 
eg,~\cite{kn:alon,kn:fomin2,kn:zheng2,kn:zheng}). 

\bigskip 

Let $G=(V,E)$ be a graph and let $f: V \rightarrow S$ be a 
function that assigns some integer from $S \subseteq \mathbb{Z}$ 
to every vertex of $G$. For a subset $W \subseteq V$ we 
write 
\[f(W)=\sum_{x \in W}\; f(x).\] 
The function $f$ is 
a domination function if 
for every vertex $x$, $f(N[x]) > 0$, where $N[x]=\{x\} \cup N(x)$ 
is the closed neighborhood of $x$.
The \underline{weight 
of $f$} is defined as the value $f(V)$. 

\bigskip 
 
In this manner, the ordinary domination problem is described 
by a domination function that assigns a value of $\{0,1\}$ 
to each element of $V$. A signed domination function 
assigns a value of $\{-1,1\}$ to each vertex $x$. The minimal 
weight of a dominating and signed dominating function 
are denoted by $\gamma(G)$ and $\gamma_s(G)$. 
In this paper we look at the minus-domination problem. 

\begin{definition}
Let $G=(V,E)$ be a graph. A function $f: V \rightarrow \{-1,0,1\}$ 
is a minus-domination function if $f(N[x]) > 0$ for every vertex 
$x$. 
\end{definition}

\bigskip 

In the minus-domination problem one tries to 
minimize the weight of a minus-domination 
function. The minimal weight of a minus-domination 
function is denoted as $\gamma^{-}(G)$. Notice that the 
weight may be negative. For example, consider a $K_4$ 
and add one new vertex for every edge, adjacent to the 
endpoints of that edge. Assign a value $1$ to every vertex of 
the $K_4$ and assign a value $-1$ to each of the six other vertices. 
This is a valid signed-domination function and its weight 
is $-2$. 

\bigskip 
 
The problem to determine the value of 
$\gamma^{-}(G)$ is NP-complete, even when restricted to 
bipartite graphs, chordal graphs and planar graphs with maximal degree 
four~\cite{kn:damaschke,kn:dunbar3}. 
Sharp bounds for the 
minimum weight are obtained in, eg,~\cite{kn:matousek}. 

Damaschke shows that, unless $P=NP$, 
the value of $\gamma^{-}$ cannot 
be approximated in polynomial time 
within a factor $1+\epsilon$, for some $\epsilon >0$, 
not even for graphs with all degrees at most 
four~\cite[Theorem 3]{kn:damaschke}. 

Famous open problems are the complexity of the minus-domination 
problem for splitgraphs and for strongly chordal graphs. 
In this paper we settle these questions. 

\section{Planar graphs}
\label{section planar}

Determining the smallest weight of a minus-dominating function 
is NP-complete, even when restricted to 
planar graphs~\cite{kn:dunbar3}. 

\bigskip 
 
Let $G=(V,E)$ be a graph and let $f:V \rightarrow S$ 
be a domination function. 
Following Zheng et al. we define the \underline{size of $f$} 
as the number of vertices $x \in V$ with $f(x) > 0$. We denote the 
size of a minus-dominating function $f$ as $\size(f)$.  

\bigskip 

Consider signed-domination functions of size at most $k$. 
It is easy to see that $|V(G)| = O(k^2)$ (see~\cite{kn:zheng2}).  
It follows that the signed domination problem parameterized 
by the size is fixed-parameter tractable. This is not so 
clear for the minus domination problem. For example, consider a 
star and assign to the center a value of $1$ and to 
every leaf a value of zero. 
This is a valid minus-domination function with size $1$ but the 
number of vertices is unbounded. 

\bigskip 

\begin{theorem}
For planar graphs the minus-domination problem,  
parameterized by the size, is fixed-parameter tractable. 
\end{theorem}
\begin{proof}
Let $f: V \rightarrow \{-1,0,1\}$ be a minus-domination 
function. Let 
\[D=\{\;x\;|\; x \in V \quad\text{and}\quad f(x)=1\;\}.\] 
Then $D$ is a dominating set in $G$. 
It follows that, for all graphs $G$,  
\[\gamma^{-}(G) \leq \gamma(G) 
\leq \min\;\{\;\size(f)\;|\; \text{f is a minus-dominating function}\;\}.\] 

\medskip 

\noindent
The first subexponential fixed-parameter algorithm for 
domination in planar graphs appeared in~\cite{kn:alber}. 
In this paper the authors prove that, 
if $G$ is a planar graph with $\gamma(G) \leq k$, then the 
treewidth of $G$ is $O(\sqrt{k})$. Using a treedecomposition 
of bounded treewidth one can solve the domination problem 
in $O(2^{15.13 \sqrt{k}}\cdot k + n^3+k^4)$ time 
(or conclude that $\gamma(G) > k$). 
The results were generalized to some nonplanar classes of graphs 
by Demaine, et al. 

\medskip  

\noindent
The minus-domination problem with size bounded by $k$ can be 
formulated in monadic second-order logic.  
By Courcelle's theorem, any such problem can be solved in linear time 
on graphs of bounded treewidth 
(see, eg,~\cite{kn:kloks,kn:kloks2}).  
This proves the theorem.
\qed\end{proof}

\subsection{$d$-Degenerate graphs}

\begin{definition}
A graph is $d$-degenerate if each of its induced subgraphs 
has a vertex of degree at most $d$. 
\end{definition}

Graphs with bounded degeneracy contain, eg, graphs that are 
embeddable on some fixed surface, families of graphs that 
exclude some minor, graphs of bounded treewidth, etc.~\cite{kn:thomason}. 

\bigskip 

In this section we show that, for each fixed $d$, 
the minus domination problem, 
parameterized by the number $k$ of vertices that receive a 
$1$, is fixed-parameter tractable for $d$-degenerate graphs. 

\bigskip 

In this section, when considering 
a partition of the vertices, we allow 
that some parts of the partition are empty. 

\bigskip 
 
In the minus domination problem one searches for a 
partition of the vertices into three parts, say red, white 
and blue. The red vertices are assigned $-1$, white are $1$ 
and blue are $0$. 
Zheng et al. proved a lemma similar to the one below for 
the signed domination problem in~\cite[Theorem 2]{kn:zheng} 
and~\cite[Lemma 6]{kn:zheng2}. 

\begin{lemma}
\label{lm zheng}
Assume that $G=(V,E)$ has a minus-dominating function with 
size at most $k$. 
Let $R$, $W$ and $B$ be the coloring 
of the vertices into red, white and blue, defined by this 
minus-domination function. Then 
\[|W \cup R| =O(k^2).\]
\end{lemma}
\begin{proof}
By assumption, the minus-domination function colors 
at most $k$ vertices white. Consider the subgraph $G^{\prime}$ 
induced 
by the red and white vertices. Consider a vertex $x$ of 
$G^{\prime}$. Then at least half of its neighbors is 
colored white, otherwise its closed neighborhood 
has weight at most zero. Since there are at most $k$ white 
vertices, each vertex of $G^{\prime}$ has degree 
less than $2k$. 

\medskip 

\noindent
Notice also that each red vertex has at 
least two white neighbors. Since there are only $k$ white 
vertices, and each white vertex has degree less than $2k$, 
the number of red vertices is less than $2k^2$. 
This proves the lemma. 
\qed\end{proof}
 
\bigskip 

For algorithmic purposes one usually considers the 
following generalization of the domination problem. 
Consider graphs of which each vertex is either colored black 
or white. In the parameterized black-and-white domination 
problem the objective is to find a set $D$ of at most $k$ vertices 
such that 
\[\boxed{\text{for each black vertex $x$, $N[x] \cap D \neq \es$.}}\]  
Obviously, the domination problem is a special case, in 
which each vertex is black. 

\bigskip 

For the minus-domination problem we describe 
an algorithm for a black-and-white 
version, where the vertices with a $0$ or $-1$ are black 
and such that each closed neighborhood of a black vertex 
has a positive weight. 
To see that this solves the minus-domination problem, 
just consider the case where all vertices are black. 

\bigskip 

Alon and Gutner prove, in their seminal paper,  
that the domination problem is fixed-parameter 
tractable for $d$-degenerate graphs~\cite{kn:alon}. 
The main ingredient of their paper is the following lemma. 

\begin{lemma}
\label{lm alon}
Let $G=(V,E)$ be a $d$-degenerate black-and-white colored graph. 
Let $B$ and $W$ be the set of black and white vertices. 
If $|B| > (4d+2)k$ then the set 
\[\Omega=\left\{\;x\; \mid \; x \in V \quad\text{and}\quad 
|N[x] \cap B| \geq \frac{|B|}{k}\;\right\} 
\quad\text{satisfies}\quad |\Omega| \leq (4d+2)k.\]  
\end{lemma}

\bigskip 

To prove that the minus-domination problem, parameterized 
by the size, is fixed-parameter tractable for $d$-degenerate 
graphs, we adapt the proof of~\cite[Theorem 1]{kn:alon}. 

\begin{theorem}
For each $d$ and $k$, there exists a linear  
algorithm for finding a minus-domination of size at most $k$ 
in a $d$-degenerate black-and-white graph, if such a set exists. 
\end{theorem}
\begin{proof}
Let $B$ and $W$ be the set of black and white vertices. 
First assume that $|B| \leq (4d+2)k$. If there is a 
minus-domination function of size at most $k$ then there 
are $k$ vertices (assigned $1$) that dominate all vertices in $B$. 

\medskip 

\noindent
The algorithm tries all possible subsets $R \subseteq B$ 
for the set of red vertices (those are assigned $-1$). 
Number the closed neighborhoods of the vertices in $R \cup B$, 
say 
\[N_1, \dots,N_t,\] 
where $t = |B \cup R| \leq (4d+2)k$. 
Define an equivalence relation on the vertices of 
$V \setminus R$ by making two vertices equivalent if 
they are contained in exactly the same subsets $N_i$. 
For each equivalence class that contains more than 
$k$ vertices which are not red, remove all of them except 
at most $k$ vertices. 
This kernelization reduces the graph to an instance $H$ 
with at most $g(k,d)$ vertices, for some function $g$. 
 
\medskip 

\noindent
Consider all subsets of $V(H)$ with at most $k$ 
vertices of which none is red. Give these vertices 
the value $1$ and the remaining vertices that are not red the 
value $0$. Check if this is a valid minus-domination. 

\medskip 

\noindent 
Now assume that $|B| > (4d+2)k$. Then, by Lemma~\ref{lm alon},  
$|\Omega| \leq (4d+2)k$. Notice that 
at least one vertex of $\Omega$ is assigned $1$ in any 
minus-domination function of size $k$. 
In that case the algorithm grows a search tree 
of size at most $(4d+2)^k \cdot k!$ before it arrives at 
the previous case (see~\cite{kn:alon}).   
\qed\end{proof}  

\section{Cographs}
\label{section cograph}

A minus domination  
with bounded size  
can be formulated in monadic second-order logic without 
quantification over subsets of edges. It follows that 
there is a linear-time algorithm to solve the problem  
for graphs of bounded treewidth or rankwidth (or cliquewidth)%
~\cite{kn:langer}. 
It is less obvious that $\gamma^-$ 
is computable for bounded rankwidth when there is 
no restriction on the size. 
In this section 
we adapt a method of Yeh and Chang to show this. 

\bigskip 

It is well-known that the graphs of rankwidth one are 
the distance-hereditary graphs. 
We first analyze the complexity of the 
minus-domination problem for the class 
of cographs. 
Cographs form a proper subclass of the 
class of distance-hereditary graphs. 

\bigskip 

We denote a path with four vertices by $P_4$. 

\begin{definition}
A cograph is a graph without induced $P_4$. 
\end{definition}

\bigskip 

Cographs are characterized by the property that each 
induced subgraph with at least two vertices is either 
a join or a union of two smaller cographs. 
It follows that cographs admit a decomposition tree 
$(T,f)$ where $T$ is a rooted binary tree and 
where $f$ is a bijection from the vertices of $G$ to 
the leaves of $T$. Each internal node is labeled as 
$\otimes$ or $\oplus$. When the label is $\otimes$ then 
all vertices of the left subtree are adjacent to all 
vertices of the right subtree. A node that is 
labeled as $\otimes$ is called a join-node. When the label is 
$\oplus$ there is no edge between vertices of the right 
and left subtree. A node that is labeled 
as $\oplus$ is called a union-node. 
One refers to a decomposition tree of 
this type as a cotree. A cotree for a cograph can be obtained 
in linear time.  

\bigskip 

\begin{theorem}
\label{thm cograph}
There exists an efficient algorithm that 
computes $\gamma^{-}$ for cographs. 
\end{theorem}
\begin{proof}
Let $G=(V,E)$ be a cograph. We assume that a cotree for $G$ 
is a part of the input. Consider a subtree $T^{\prime}$ and let 
$W \subseteq V$ be the set of vertices that are mapped to 
the leaves in $T^{\prime}$. 

\medskip 

\noindent 
For three numbers $(a,b,c)$, an $(a,b,c)$-function 
is a function $f:W \rightarrow \{-1,0,1\}$ such that 
$f$ assigns $a$ vertices the value $-1$, $b$ vertices the 
value $0$ and $c$ vertices the value $1$. Obviously, we 
have that $a+b+c=|W|$.  

\medskip 

\noindent 
For an integer $t$, let 
\begin{multline}
\label{eqn1} 
\zeta(t,a,b,c)= \max\;|\;\{\;x\;\mid\; x \in W \quad  
\text{and}\quad f(N[x] \cap W) + t > 0 \quad \text{and}\\
\text{where $f$ is an $(a,b,c)$-function}\;\}\;|.
\end{multline}
When the set is empty we let $\zeta(t,a,b,c)=- \infty$. 

\medskip 

\noindent 
Notice that a minus-domination function 
with minimum weight can be computed when $\zeta$ is 
known for the root node, that is, when $W=V$. Namely, 
\begin{equation}
\label{eqn2}
\gamma^{-}(G)= 
\min\;\{\;-a+c\;\mid\; a+b+c=n \quad\text{and}\quad
\zeta(0,a,b,c)=n\;\}.
\end{equation}

\medskip 

\noindent 
We show how the values $\zeta(t,a,b,c)$ can be computed. 
Assume that $G$ is the union of two cographs 
$G_1=(V_1,E_1)$ and $G_2=(V_2,E_2)$. 
We denote the $\zeta$-values for $G_1$ and $G_2$ by 
$\zeta_1$ and $\zeta_2$.  
Then 
\begin{multline}
\label{eqn3}
\zeta(t,a,b,c) = \max \; \{\;\zeta_1(t,a_1,b_1,c_1)+\zeta_2(t,a_2,b_2,c_2) \\  
\text{where}\quad a_1+a_2=a \quad b_1+b_2=b \quad c_1+c_2=c\;\}.
\end{multline}

\medskip 

\noindent 
Now assume that $G$ is the join of $G_1$ and $G_2$. 
Then 
\begin{multline}
\label{eqn4}
\zeta(t,a,b,c)=\max\;\{\; \zeta_1(t-c_2+a_2,a_1,b_1,c_1)+ 
\zeta_2(t-c_1+a_1,a_2,b_2,c_2) \\
\text{where}\quad a_1+a_2=a \quad 
b_1+b_2=b \quad c_1+c_2=c\;\}.
\end{multline}
This proves the theorem. 
\qed\end{proof}

\bigskip 

\begin{remark}
Notice that complete multipartite graphs are cographs. 
Formulas for the signed and minus domination number of 
complete multipartite graphs appear 
in a recent paper by H.~Liang. 
\end{remark}

\bigskip 

By similar methods we obtain 
a polynomial algorithm for minus domination 
on distance-hereditary graphs. 
For brevity 
we put the proof of the next theorem in an appendix. 

\begin{theorem}
There exists a polynomial algorithm that computes 
$\gamma^-$ for distance-hereditary graphs. 
\end{theorem}

\bigskip 

\begin{remark}
It is not hard to see that similar results can be derived for 
graphs of bounded rankwidth, that is, $\gamma^-$ is computable 
in polynomial time for graphs of bounded rankwidth 
(see, eg,~\cite{kn:kloks2}). The rankwidth 
appears as a function in the exponent of $n$. 
Graphs of bounded treewidth 
are contained in the class of bounded rankwidth and so a similar statement 
holds for graphs of bounded treewidth. At the moment we do not believe  
that there is a fixed-parameter algorithm, parameterized 
by treewidth or rankwidth, to compute $\gamma^-$.    
The results of~\cite[Section 4.2]{kn:zheng} seem wrong.\footnote{We 
communicated with the authors of~\cite{kn:zheng} and 
our ideas about it are now in agreement.} 
\end{remark}

\section{Strongly chordal graphs}
\label{section strongly chordal}

The minus domination problem is NP-complete for 
chordal graphs. 
In this section we show that the problem can be solved in 
polynomial time for strongly chordal graphs. 

\bigskip 

A graph is chordal if it has no induced cycle of length 
more than three. 
A chord in a cycle is an edge that runs between two vertices 
that are not consecutive in the cycle. 
Let $C=[x_1,\dots,x_{2k}]$ be an even cycle of 
length $2k$. A chord $\{x_i,x_j\}$ in $C$ is odd 
if the distance in $C$ between $x_i$ and $x_j$ is odd. 

\begin{definition}
A chordal graph $G$ is strongly chordal if each cycle in $G$ 
of even length at least 6 has an odd chord. 
\end{definition}

There are many characterizations of strongly chordal 
graphs~\cite{kn:farber2,kn:kloks3}. 
Perhaps the best known examples of strongly 
chordal graphs are the interval graphs. 

\bigskip 

In strongly chordal graphs the domination number is 
equal to the 2-packing 
number 
(see, eg,~\cite[Theorem 7.4.4]{kn:scheinerman}). 
It follows that the domination number for strongly chordal 
graphs is polynomial~\cite{kn:farber2}. 

\bigskip 

\begin{theorem}
The minus domination problem for strongly 
chordal graphs can be solved in $O(\min\;\{n^2, m \log n\})$ time. 
Here $n$ is the number of vertices and $m$ is the number of edges 
of the graph.  
\end{theorem}
\begin{proof}
Farber describes a linear programming formulation 
for the domination problem. In this linear programming 
formulation we can change the variables from $x_i$ to $z_i=x_i+1$. 
This changes the constraints $-1 \leq x_i \leq 1$ into 
$0 \leq z_i \leq 2$. The linear program becomes 
\begin{align*}
\text{Minimize} & \quad \sum_{i=1}^n \; z_i & \\
\text{subject to} & \quad\sum_{i \in N[k]}\; z_i \geq b_k & 
\text{for each $k$}\\
\text{and} & \quad 0 \leq z_i \leq 2 & \text{for each $i$}.
\end{align*}
In our case, the variable $b_k$ is equal to $|N[k]|+1$. 

\medskip 

\noindent 
The closed neighborhood matrix of a strongly chordal graph is 
totally balanced.  By~\cite{kn:fulkerson,kn:hoffman,kn:kolen} (see 
also, eg,~\cite[Theorem A.3.4]{kn:scheinerman}) 
the 
integer program and its linear relaxation have the same value.

\medskip 

\noindent 
To deal with the constraints $z_i \leq 2$ we write the LP as 
\begin{align*}
\text{Minimize} & \quad \mathbf{j}^T \cdot \mathbf{z} & \\
\text{subject to} & \quad 
\begin{pmatrix} A \\ -I \end{pmatrix} \mathbf{z} \geq 
\begin{pmatrix} \mathbf{b} \\ -2 \cdot \mathbf{j} \end{pmatrix} & 
\text{and} & \quad \mathbf{z} \geq \mathbf{0}. 
\end{align*}
Here, the matrix $A$ is the closed 
neigborhood matrix, and  
the vector $\mathbf{b}$ is equal to  
\[\mathbf{b}=\mathbf{j}+ A \mathbf{j}.\]  

\medskip 

\noindent 
The dual of this LP is 
\begin{align*}
\text{Maximize} & \quad \mathbf{b}^T \cdot \mathbf{y_1}-
2\mathbf{j}^T \cdot \mathbf{y_2}
\\
\text{subject to} &  \quad A\mathbf{y_1} \leq \mathbf{j}+\mathbf{y_2}  
\quad \text{and} & \quad \mathbf{y_1} \geq \mathbf{0} \quad\text{and}\quad 
\mathbf{y_2} \geq \mathbf{0}. 
\end{align*}
Notice that 
\[y_{2,k}=\max\;\{\;0,\;-1+\sum_{i \in N[k]}\;y_{1,i}\;\} \quad 
\text{for all $k$.}\] 

\medskip 

\noindent 
The complementary slackness conditions are as follows. 
\begin{eqnarray*}
y_{1,k} > 0 &\quad\Rightarrow\quad& \sum_{i \in N[k]} z_i=1+|N[k]| \\
\sum_{i \in N[k]}y_{1,i} > 1 & \quad\Rightarrow\quad & z_k=2, \quad\text{and}\\
z_k > 0 & \quad\Rightarrow\quad& \sum_{i \in N[k]} y_{1,i} \geq 1.
\end{eqnarray*} 

\medskip 

\noindent
Solving the linear problem can be done in $O(n^{3.5} \log n)$. 
Farber's method can be used to bring it down to $O(n^2)$ 
or $m \log n)$, which is the time needed to compute a 
strong elimination ordering. We omit the details; see Remark~\ref{rem lee}. 
\qed\end{proof} 

\begin{remark}
When $G$ is strongly chordal then $G^2$ is that also~\cite{kn:lubiw}. 
A simple vertex of $G$ is simplicial in $G^2$. The weighted $2$-packing 
problem in $G$ asks for the maximal weight independent set in $G^2$. 
This can be solved in linear time~\cite{kn:frank}. 
It uses the fact that in any chordal 
graph, with integer weights on the vertices, the maximal weight of 
an independent set equals the minimal number of cliques that have the 
property that every vertex is covered at least as many times by 
cliques as its weight.   
\end{remark}
  
\begin{corollary}
The exists a linear-time algorithm that solves 
minus domination on interval graphs. 
\end{corollary}

\begin{remark}
\label{rem lee}
After the publication of our draft on arXiv, one of the authors 
of their paper, quoted in the footnote, drew our attention to their result. 
The authors claim a 
linear algorithm for minus domination on 
strongly chordal graphs. (Here, they assume that a strong 
elimination ordering is a part of the input).%
~\footnote{
C.~Lee and M.~Chang, Variations of $Y$-dominating functions 
on graphs, {\em Discrete Mathematics\/} {\bf 308} (2008), pp.~4185--4204.}
\end{remark}

\section{Splitgraphs}
\label{section split}

In this section we show that the minus-domination 
problem is NP-complete for splitgraphs. 
We reduce the $(3,2)$-hitting set problem to the minus-domination 
problem. The $(3,2)$-hitting set problem is defined 
as follows (see, eg,~\cite{kn:mellor}). 

\begin{description}
\item[Instance:] Let $\mathcal{C}$ be a collection of sets, each 
containing exactly three elements from a universe $U$. 
\item[Question:] Find a smallest set $U^{\prime} \subseteq U$ 
such that for each $C \in \mathcal{C}$, 
\[|C \cap U^{\prime}| \geq 2.\] 
\end{description}

\begin{lemma}
The $(3,2)$-hitting set is NP-complete. 
\end{lemma}
\begin{proof}
The reduction is from vertex cover, ie, $(2,1)$-hitting set. 
The $(2,1)$-hitting set is defined similar as above, except that 
in this case every subset has two elements and the problem 
is to find a subset $U^{\prime}$ which hits every subset 
at least once. 

\medskip 

\noindent
Consider an instance of $(2,1)$-hitting set. Let 
$\mathcal{C}$ be the collection of $2$-element subsets 
of a universe $U$. 
Add four vertices to the universe, say $\alpha$, $\beta$, 
$\gamma$ and $\delta$. Add $\alpha$ to every subset of $\mathcal{C}$ 
and add two subsets, $\{\alpha,\beta,\gamma\}$ and 
$\{\alpha,\beta,\delta\}$. We claim that any solution 
of this $(3,2)$-hitting set problem has $\alpha$ in the 
hitting set. If not, then $\{\beta,\gamma,\delta\}$ is a 
subset of the $(3,2)$-hitting set. In that case we may 
replace the elements $\beta$, $\gamma$ and $\delta$ with 
$\alpha$ and $\beta$. Then we obtain a $(3,2)$-hitting set 
with fewer elements. 

\medskip 

\noindent 
Thus, we may assume that $\alpha$ is in the $(3,2)$-hitting set. 
But now the problem is equivalent to the $(2,1)$-hitting set, 
since every adapted subset contains $\alpha$. 
\qed\end{proof}

\begin{theorem}
\label{thm NP-c}
The minus-domination problem is NP-complete 
for splitgraphs. 
\end{theorem}
\begin{proof}
Consider an instance of the $(3,2)$-hitting set. 
We first construct a splitgraph where $U$ is the clique and where each 
element $C \in \mathcal{C}$ is a vertex of the independent set, 
and adjacent exactly to the three vertices of $C$ in the clique. 
Next, we extend 
the splitgraph by adding auxiliary vertices in the clique and the 
independent set, respectively.  Precisely, we add a set $X$ of 
$|U| + |\mathcal{C}| + 1$ vertices in the clique, 
and for each vertex $x$ in $X$, we 
add a distinct vertex $x^{\prime}$ 
in the independent set that connects with $x$.
This completes the description of the splitgraph. Let $V$ be the set 
of vertices of this graph, that is 
\[V=U \cup X \cup \{\;x^{\prime}\;|\; x \in X\;\} \cup \mathcal{C}.\] 

\medskip 

\noindent 
Consider a minus-domination function $f$ of minimal weight. 
Notice that, we may assume that for each vertex $x$ in $X$, $f(x) = 1$.  
Otherwise, by considering the closed neighborhood $N[x^{\prime}]$, 
we require $f(x^{\prime}) + f(x) > 0$,
so that $f(x^{\prime}) = 1$ and $f(x) = 0$;  
in such a case, we can reset $f(x^{\prime})$ as $0$ 
and $f(x)$ as $1$, while maintaining validity 
(i.e., positive total weight for each close neighborhood) 
and optimality (ie, minimum total weight) of the assignment. 

\medskip 

\noindent 
Notice that for any function $f: V \rightarrow \{-1,0,1\}$ we have 
that 
\[\forall_{x \in X}\; f(x)=1 \;\Rightarrow\; \forall_{u \in U}\; f(N[u])>0\] 
no matter what values  
the vertices $u \in U$  
or $C \in \mathcal{C}$ 
are assigned.  

\medskip 

\noindent
We may now, further assume that for each 
vertex $C$ in the independent set  
\begin{equation}
\label{eqn7}
f(C) \leq \min\;\{\;f(u)\;|\; u \in N(C)\;\}.
\end{equation}
If this were not the case, then we could swap the value $f(C)$ 
with the value of a vertex in $N(C)$ and obtain a minus-domination 
function of at most the same weight, satisfying~\eqref{eqn7}. 
Note that after the change, we cannot have $f(N(C)) = 0$.

\medskip 

\noindent 
We claim that there is a domination function of minimal 
weight with $f(C)=-1$ for every $C \in \mathcal{C}$. 
To see that, consider the following cases. 
If $f(N(C)) =3$, then we have $f(C)=-1$. 
If $f(N(C))=2$, then $N(C)$ has two ones and one zero. 
Also in that case we have $f(C)=-1$. 
The only case that is left is where $N(C)$ contains 
one $1$ and two zeroes and $f(C)=0$. In that case we may 
change the value of a zero in $N(C)$ to one, and $f(C)$ to 
$-1$. Repeated application of this type of exchange 
produces a minus domination function of the same 
weight and satisfying the claim. 

\medskip 

\noindent
So, we may assume that for $C \in \mathcal{C}$, $f(C)=-1$ 
and that for each vertex $u \in U$, $f(u) \in \{0,1\}$. 
Since $f(C)=-1$, the minus-domination function 
has at least two plus ones in $N(C)$. 

\medskip 

\noindent
This proves the theorem.
\qed\end{proof}

\subsection{Minus domination is not FPT}

Consider the following problem. 
\begin{description}
\item[Instance:] A graph $G$. 
\item[Question:] Does $G$ have a minus domination of weight at most $0$? 
\end{description}
Following Hattingh et al.,  
we call this `the zero minus-domination problem.' 

Consider the graph $L$ in Figure~\ref{fig L}. 
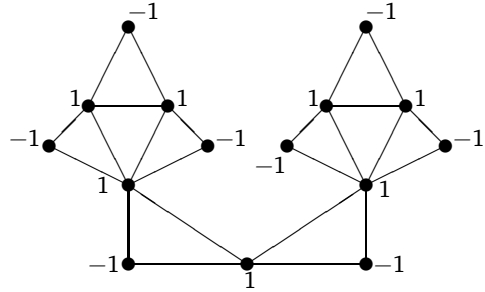
\begin{figure}
\setlength{\unitlength}{1.50pt}
\begin{center}
\begin{picture}(110,70)
\put(10,40){\circle*{3.0}}
\put(20,50){\circle*{3.0}}
\put(30,30){\circle*{3.0}}
\put(30,70){\circle*{3.0}}
\put(40,50){\circle*{3.0}}
\put(50,40){\circle*{3.0}}
\put(10,40){\line(2,-1){20}}
\put(10,40){\line(1,1){10}}
\put(50,40){\line(-2,-1){20}}
\put(50,40){\line(-1,1){10}}
\put(30,30){\line(-1,2){10}}
\put(30,30){\line(1,2){10}}
\put(20,50){\line(1,0){20}}
\put(30,70){\line(-1,-2){10}}
\put(30,70){\line(1,-2){10}}

\put(70,40){\circle*{3.0}}
\put(80,50){\circle*{3.0}}
\put(90,30){\circle*{3.0}}
\put(90,70){\circle*{3.0}}
\put(100,50){\circle*{3.0}}
\put(110,40){\circle*{3.0}}
\put(70,40){\line(2,-1){20}}
\put(70,40){\line(1,1){10}}
\put(110,40){\line(-2,-1){20}}
\put(110,40){\line(-1,1){10}}
\put(90,30){\line(-1,2){10}}
\put(90,30){\line(1,2){10}}
\put(80,50){\line(1,0){20}}
\put(90,70){\line(-1,-2){10}}
\put(90,70){\line(1,-2){10}}

\put(30,10){\circle*{3.0}}
\put(60,10){\circle*{3.0}}
\put(90,10){\circle*{3.0}}
\put(30,10){\line(0,1){20}}
\put(90,10){\line(0,1){20}}
\put(30,10){\line(1,0){60}}
\put(60,10){\line(-3,2){30}}
\put(60,10){\line(3,2){30}}

\put(0,40){$-1$}
\put(15,50){$1$}
\put(30,72){$-1$}
\put(52,40){$-1$}
\put(42,50){$1$}
\put(22,28){$1$}

\put(62,33){$-1$}
\put(75,50){$1$}
\put(90,72){$-1$}
\put(102,50){$1$}
\put(112,40){$-1$}
\put(93,27){$1$}

\put(20,8){$-1$}
\put(59,4){$1$}
\put(92,8){$-1$}

\end{picture}
\caption{The graph $L$. It has $\gamma^{-}(L)=-1$.}
\label{fig L}
\end{center}
\end{figure}

\begin{lemma}
The graph $L$ has minus-domination weight $\gamma^{-}(L)=-1$. 
The minus-domination function that achieves this weight 
is unique; it is the one depicted in Figure~\ref{fig L}. 
\end{lemma}

\begin{theorem}
The zero minus-domination problem is NP-complete. 
\end{theorem}
\begin{proof}
Let $H$ be a graph and let $G$ be the union of $H$ 
and $k$ disjoint copies of $L$. 
Obviously 
\[\gamma^{-}(G)=\gamma^{-}(H)+k\cdot \gamma^{-}(L)=\gamma^{-}(H)-k.\] 
It follows that $\gamma^{-}(G) \leq 0$ if and only if $\gamma^{-}(H) \leq k$. 
By Theorem~\ref{thm NP-c}, given a graph $H$ and a positive $k$ 
it is NP-complete to decide whether $\gamma^{-}(H) \leq k$. 
\qed\end{proof}

\bigskip 

\begin{theorem}
The minus-domination problem is not fixed-parameter 
tractable, unless $P=NP$. 
\end{theorem}
\begin{proof}
Assume there exists an algorithm which runs in time 
$O(f(k)\cdot n^c)$ and that determines whether a graph $G$ has a 
minus domination of weight at most $k$. 
Then the zero minus-domination problem would be solvable 
in polynomial time.  
\qed\end{proof}

\appendix 

\section{Distance-hereditary graphs}

Distance-hereditary graphs are the graphs of rankwidth one 
(see, eg,~\cite{kn:kloks2}). They were introduced in 1977 by Howorka as 
those graphs in which, for every pair of nonadjacent vertices, 
all the cordless paths between them have the same length. 
They have a decomposition tree $(T,f)$ where $T$ is a rooted 
binary tree and $f$ is a bijection from the vertices to the leaves of 
$T$. For each branch, the `twinset' of that branch 
is defined as those vertices in 
the leaves that have neighbors in leaves outside that branch. 
Each twinset induces a cograph. Each internal node of $T$ is 
labeled as $\oplus$ or $\otimes$. When the label is $\otimes$ 
then all the vertices in the twinset of the left branch 
are adjacent to all the vertices in the twinset of the right branch. 
When the label is $\oplus$ there are no edges between vertices mapped to 
different branches. 
The twinset of a parent is either empty, or the twinset of one 
of the two children or the union of the twinsets at the two 
children. 

\begin{theorem}
\label{thm DH}
There exists a polynomial algorithm that computes 
$\gamma^-$ for distance-hereditary graphs. 
\end{theorem}
\begin{proof}
Consider a branch $B$ and let $W$ be the set of vertices that 
are mapped to leaves of $B$. Let $Q$ be the twinset of $B$, 
that is, the set of vertices in $W$ that have neighbors in 
$V \setminus W$. 

\medskip 

\noindent
An $(a,b,c)$-function is a function $f:W \rightarrow \{-1,0,1\}$ 
such that $f$ assigns to $a$ vertices of $Q$ the value $-1$, 
to $b$ vertices of $Q$ the value $0$ and to $c$ vertices of $Q$ 
the value $1$. Furthermore, 
\begin{equation}
\label{eqn5}
\text{for all $x \in W \setminus Q$}\quad f(N[x]) > 0.
\end{equation}

\medskip 

\noindent
For an integer $t$ let $\zeta(t,a,b,c)$ be defined as  
\begin{multline}
\label{eqn6} 
\zeta(t,a,b,c)= \max\;|\;\{\;x\;\mid\; x \in Q \quad  
\text{and}\quad f(N[x] \cap W) + t > 0 \quad \text{and}\\
\text{where $f$ is an $(a,b,c)$-function}\;\}\;|.
\end{multline}
It is a nice, easy exercise to see that the arguments given 
in the proof of Theorem~\ref{thm cograph} extend to show 
that these definitions lead to an efficient computation of 
$\gamma^-$ for distance-hereditary graphs. For brevity we omit the 
details. 
\qed\end{proof} 


\begin{thebibliography}{99}

\bibitem{kn:alber}Alber,~J., H.~Bodlaender, H.~Fernau, T.~Kloks 
and R.~Niedermeier, 
Fixed-parameter algorithms for dominating set and related 
problems on planar graphs, 
{\em Algorithmica\/} {\bf 33} (2002), pp.~461--493. 

\bibitem{kn:alon}Alon,~N. and S.~Gutner, 
Linear time algorithms for finding a dominating set 
of fixed size in degenerated graphs, 
{\em Algorithmica\/} {\bf 54} (2009), pp.~544--556. 

\bibitem{kn:damaschke}Damaschke,~P., 
Minus domination in small-degree graphs, 
{\em Proceedings WG'98\/}, Springer-Verlag, LNCS 1517 (1998), pp.~17--25. 

\bibitem{kn:dunbar3}Dunbar,~J., W.~Goddard, S.~Hedetniemi, 
M.~Henning and A.~McRae, 
The algorithmic complexity of minus domination in graphs, 
{\em Discrete Applied Mathematics\/} {\bf 68} (1996), 
pp.~73--84. 
 
\bibitem{kn:farber2}Farber,~M., 
Domination, independent domination, and duality in strongly 
chordal graphs, 
{\em Discrete Applied Mathematics\/} {\bf 7} (1984), pp.~115--130. 

\bibitem{kn:fomin2}Fomin,~F., D.~Lokshtanov, S.~Saurabh and 
D.~Thilikos, 
Linear kernels for (connected) dominating set on graphs with 
excluded topological subgraphs, 
{\em Proceedings STACS'13\/}, Schloss Dagstuhl-Leibniz-Zentrum 
f\"ur Informatik, LPIcs {\bf 20} (2013), pp.~92--103. 

\bibitem{kn:frank}Frank,~A., 
Some polynomial algorithms for certain graphs and 
hypergraphs, 
{\em Proceedings $5^{\mathrm{th}}$ British Combinatorial Conference 
1975\/}, (Eds. C.~Nash-Williams and J.~Sheehan), 
Congressus Numeratium XV, pp.~211--226. 
 
\bibitem{kn:fulkerson}Fulkerson,~D., A.~Hoffman and R.~Oppenheim, 
On balanced matrices, 
{\em Mathematical Programming Study\/} {\bf 1} (1974), pp.~120--132. 
 
\bibitem{kn:hoffman}Hoffman,~A., A.~Kolen and M.~Sakarovitch, 
Totally-balanced and greedy matrices, 
{\em Siam Journal on Algebraic and Discrete Methods\/} {\bf 6} 
(1985), pp.~721--730. 

\bibitem{kn:kloks}Kloks,~T., 
{\em Treewidth -- Computations and Approximations\/}, 
Springer-Verlag, LNCS 842, 1994. 

\bibitem{kn:kloks3}Kloks,~T., C.~Liu and S.~Poon, 
Feedback vertex set on chordal bipartite graphs. 
Manuscript on arXiv: 1104.3915, 2012. 

\bibitem{kn:kloks2}Kloks,~T. and Y.~Wang, 
{\em Advances in graph algorithms\/}. Manuscript 2013. 

\bibitem{kn:kolen}Kolen,~A., 
{\em Location problems on trees and in the rectilinear plane\/}, 
PhD Thesis, Mathematisch centrum, Amsterdam, 1982. 

\bibitem{kn:langer}Langer,~A., P.~Rossmanith and S.~Sikdar, 
Linear-time algorithms for graphs of bounded rankwidth: 
A fresh look using game theory, 
{\em Proceedings TAMC'11\/}, Springer-Verlag, LNCS 6648 
(2011), pp.~505--516. 

\bibitem{kn:lubiw}Lubiw,~A., 
{\em $\Gamma$-Free matrices\/}, 
Master's Thesis, University of Waterloo, Canada, 1982. 

\bibitem{kn:matousek}Matou\v{s}ek,~J., 
Lower bound on the minus-domination number, 
{\em Discrete Mathematics\/} {\bf 233} (2001), pp.~361--370. 

\bibitem{kn:mellor}Mellor,~A., E.~Prieto, L.~Mathieson and P.~Moscato, 
A kernelisation approach for multiple $d$-hitting set 
and its application in optimal multi-drug therapeutic 
combinations, 
{\em PLoS ONE\/} {\bf 5} (2010), e13055. 

\bibitem{kn:sawada}Sawada,~J. and J.~Spinrad, 
{F}rom a simple elimination ordering to a strong 
elimination ordering in linear time, 
{\em Information Processing Letters\/} {\bf 86} (2003), 
pp.~299--302. 

\bibitem{kn:scheinerman}Scheinerman,~E. and D.~Ullman, 
{\em Fractional graph theory\/}, Wiley, 1997. 

\bibitem{kn:thomason}Thomason,~A., 
The extremal function for complete minors,
{\em Journal of Combinatorial Theory, Series B\/} {\bf 81} (2001), 
pp.~318--338. 

\bibitem{kn:zheng2}Zheng,~Y., J.~Wang and Q.~Feng, 
Kernelization and lowerbounds of the signed domination problem, 
{\em Proceedings FAW-AAIM'13\/}, Springer-Verlag, LNCS 7924 (2013), 
pp.~261--271. 

\bibitem{kn:zheng}Zheng,~Y., J.~Wang, Q.~Feng and J.~Chen, 
FPT results for signed domination, 
{\em Proceedings TAMC'12\/}, Sprinver-Verlag, LNCS 7287 
(2012), pp.~572--583. 
 
\end{thebibliography}
\end{document}